\documentclass{article}
\usepackage{fullpage}
\usepackage{epsfig}
\usepackage{graphics}
\usepackage{latexsym}
\usepackage{amsmath}
\usepackage{amsfonts}
\usepackage{amssymb}
\usepackage{mathrsfs}
\usepackage{pifont}
\usepackage{amsthm}
\usepackage[usenames,dvipsnames,svgnames,table]{xcolor}
\usepackage{rotating}
\usepackage{authblk}

\newtheorem{theorem}{Theorem}
\newtheorem{lemma}[theorem]{Lemma}

\newtheorem{algorithm}[theorem]{Mechanism}

\makeatletter
\def\section{\@startsection {section}{1}{\z@}{-3.5ex plus -1ex minus
 -.2ex}{2.3ex plus .2ex}{\large\bf}}
\makeatother

\def\bfm#1{\mbox{\boldmath$#1$}}

\def\0{\bfm 0}

\DeclareMathAlphabet{\mathpzc}{OT1}{pzc}{m}{it}

\title{\bf Copula-based Randomized Mechanisms for Truthful Scheduling on Two  Unrelated Machines}
\author[1]{Xujin Chen}
\author[2]{Donglei Du}
\author[3]{Luis F. Zuluaga}
\affil[1]{Institute of Applied Mathematics, AMSS, Chinese Academy of Sciences

Beijing 100190, China.

{\tt xchen@amss.ac.cn}}
\affil[2]{Faculty of Business Administration, University of New Brunswick,

Fredericton NB Canada E3B 9Y2.

{\tt ddu@unb.ca}}
\affil[3]{Department of Industrial and Systems Engineering, Lehigh University,

Bethlehem, PA, USA 18015.

{\tt luis.zuluaga@lehigh.edu} }

\begin{document}
 \date{}
\maketitle

\begin{abstract} We design a Copula-based generic randomized truthful mechanism for scheduling on two unrelated machines with approximation ratio within $[1.5852, 1.58606]$, offering an improved upper bound for the two-machine case.  Moreover, we provide an upper bound $1.5067711$ for the two-machine two-task case, which is almost tight in view of the lower bound of 1.506  for the scale-free truthful mechanisms \cite{l09}. Of independent interest is the explicit incorporation of the concept of Copula in the design and analysis of the proposed approximation algorithm. We hope that techniques like this one will also prove useful in solving other problems in the future.
\end{abstract}

\noindent{\bf keywords.}\quad Algorithmic mechanism design, Random mechanism, Copula, Truthful scheduling

\openup 1.22\jot

\newcounter{my}
\newenvironment{mylabel}
{
\begin{list}{(\roman{my})}{
\setlength{\parsep}{-1mm}
\setlength{\labelwidth}{8mm}
\usecounter{my}}
}{\end{list}}

\newcounter{my2}
\newenvironment{mylabel2}
{
\begin{list}{(\alph{my2})}{
\setlength{\parsep}{-1mm} \setlength{\labelwidth}{8mm}
\setlength{\leftmargin}{6mm}
\usecounter{my2}}
}{\end{list}}

\newcounter{my3}
\newenvironment{mylabel3}
{
\begin{list}{(\alph{my3})}{
\setlength{\parsep}{-1mm}
\setlength{\labelwidth}{8mm}
\setlength{\leftmargin}{14mm}
\usecounter{my3}}
}{\end{list}}

\section{Introduction}
The main focus of this work is to offer randomized truthful mechanisms with improved approximation for {minimizing makespan on unrelated parallel machines:} $R2||C_{\max}$, a central problem extensively investigated in both the classical scheduling theory and the more recent algorithmic mechanism design initiated by  {the seminal work of} Nisan and Roenn \cite{nr01}.

Formally, {we are interested in the following  scheduling problem}: there are $n$ tasks to be processed by {$m$} machines. Machine {$i\in\{1,2,\ldots,m\}$}   takes $t_{ij}$ time to process task {$j\in\{1,2,\ldots, n\}$.} The objective is to schedule these  tasks non-preemptively on these machines to minimize the  makespan -- the latest completion time among all the tasks.
An allocation  {for the scheduling problem} 
is specified by a set of binary variables { $x_{ij}$ such that}
$x_{ij}=1$ if and only if task $j$ is allocated to machine $i$.

Different from traditional approximation algorithms for the scheduling problem, we focus on the class of monotonic algorithms defined as follows: an allocation {or a} {scheduling} algorithm is \emph{monotonic} if for any two instances of the scheduling problem $t_{ij}$ and $\widetilde{t}_{ij}$ ({$i=1,2,\ldots,m$ and $j=1,2,\ldots, n$}) differing only on  {a single} machine, the allocation  $x_{ij}$ and $\widetilde{x}_{ij}$ returned by the algorithm satisfies $\sum_{j=1}^n\left(x_{ij}-\widetilde{x}_{ij}\right)\left(t_{ij}-\widetilde{t}_{ij}\right)\le0$ for all $i=1,2,\ldots,m$.

The interest in monotonic algorithms stems from its connection to truthful mechanism design, where selfish agents maximize their profit by revealing their true private information. In this particular scheduling problem,  {a mechanism consists of two algorithms, an \emph{allocation} algorithm which allocates tasks to machines and a \emph{payment} algorithm which specifies the payment every machine receives. Each machine is a selfish} agent who knows its own processing time for every task and wants to maximize its own payoff -- the payment received minus the total execution time for the tasks allocated to it.
A mechanism is \emph{truthful} if it is a dominant strategy for each machine to reveal  its
true processing time. It is well-known that the monotonicity property above characterizes the allocation algorithm in any truthful mechanism for the scheduling problem on-hand {(see e.g., \cite{KV07}).}  
In this paper,
we are {concerned with} the approximation ratio of monotonic allocation algorithms.  {When the allocation algorithm is randomized, i.e., the binary variables $x_{ij}$ ($i=1,2,\ldots,m$, $j=1,2,\ldots,n$) output by the algorithm are random variables, we call the allocation algorithm {\em monotonic} if it is a probability distribution over a family of deterministic monotone allocation algorithms.
Every monotonic randomized allocation algorithm gives rise to a (universally) truthful mechanism~\cite{nr01}.}

As usual, the {\em approximation ratio} of an allocation algorithm is the worst-case ratio between the makespan of the allocation output by the algorithm and the optimal makespan. One fundamental open problem  {on the mechanism design for scheduling} 
is to find the exact approximation ratios $R_{\textsc{det}}$ and $R_{\textsc{ran}}$ among all monotonic deterministic and randomized allocation algorithms respectively  {\cite{nr01}}. The current best bounds are {$2.618\approx 1+\phi \le R_{\textsc{det}}\le m$} with the upper and lower bounds  established by
Nisan and Ronen \cite{nr01} and Koutsoupias and Vidali~\cite{KV07}, respectively, and $2-1/m\le R_{\textsc{ran}}\le 0.83685m$ with the upper and lower bounds  proved by
Mu'alem and Schapira \cite{mu2007setting} and Lu and Yu  \cite{ly2008-stacs}, respectively.

 {In view of the unbounded gap between the lower and upper bounds for the general $m$ machines, a lot of research efforts have been devoted to the special case of $m=2$ machines (see e.g., \cite{ckv08,l09,nr01}), which is highly nontrivial and suggests more insights for resolving the general problem.}
In this {paper}, we will focus on the  {two-machine case.}  {The deterministic approximation is exactly $2$ as shown by Nisan and Ronen \cite{nr01}}. The currently best randomized approximation ratio is shown to lie {between $1.5$ and $1.6737$. {The upper bound due to
Lu and Yu was proved by introducing a unified framework for designing truthful mechanisms~\cite{ly2008-stacs}. This} improved     Lehmann's ratio of {1.75   for Nisan and Ronen's mechanism~\cite{nr01} by 0.0763.} Later, Lu and Yu
\cite{ly08} provided an improved ratio of 1.5963, {whose proof} unfortunately is incorrect as shown in this paper later in Section~\ref{LuYu}.  Dobzinski and Sundararajan \cite{ds08} and Christodoulou et al. \cite{ckv08} independently showed that any monotonic allocation algorithm for two machines with a finite approximation ratio is {\em weakly task-independent}, meaning that, for any task, its  allocation    does not change as long as none of its {own} processing time  on machines changes. The weak task-independence is strengthened to be a {\em strong} one if the random variables $x_{ij}$ output by the allocation algorithm are independent between different tasks \cite{l09}.

In this paper, we use the concept of Copula to address the correlations among random outputs of the allocation algorithm under Lu and Yu's framework~\cite{ly2008-stacs}. Our main contribution
is to offer a Copula-based generic randomized mechanism  {for two-machine scheduling} with approximation ratio within $[1.5852, 1.58606]$,  reducing the existing best upper bound \cite{ly2008-stacs} by more than $0.0876$. 
 Moreover, we provide  an upper
bound of $1.5067711$ for the two-machine two-task case, which  improves upon the previous 1.5089 bound given in \cite{l09} and is almost tight in view of the lower bound of 1.506  for the so called scale-free monotonic allocation algorithm \cite{l09}.

To our best knowledge, we are unaware of any extant work on the explicit usage of the concept of Copula in the design and analysis of approximation {algorithms}. We hope that {techniques} like this one will also prove useful in solving other problems in the future.

The rest of the paper is organized as follows: We present the Copula-based generic randomized mechanism in Section~\ref{generic}. We then analyze the mechanism for  strongly independent tasks and weakly independent tasks  
in Section~\ref{indep} and Section~\ref{dep} respectively.  Finally, we conclude the paper with some remarks on our choice of Copula in Section \ref{sec:conclude}. The omitted details   can be found in Appendix.

\section{A generic randomized mechanism based on copula}\label{generic}
Given any real $\alpha$, we use $\alpha^+$ to denote the nonnegative number $\max\{0,\alpha\}$.
Let $F:\mathbb R_+\rightarrow[0,1]$ be a non-decreasing function satisfying $F(0)=0$ and $\lim_{x\rightarrow\infty}F(x)=1$. Write $\bar F(x)$ for $1-F(x)$. Let $ X_1, X_2,\ldots, X_n$ be $n$ dependent random variables with joint distribution function $\text{Pr}( X_1\le x_1, X_2\le x_2,\ldots, X_n\le x_n)$ given by the Clayton Copula
\begin{align}\label{defG}
G(x_1,x_2,\ldots,x_n)
=&\left[\left(\sum_{i=1}^n\sqrt[n-1]{F(x_i)}-n+1\right)^+\right]^{n-1}.
\end{align}
 It is easy to see that for any $1\le i<j\le n$, the joint distribution of $X_i$ and $X_j$ is given by
 \begin{align}
 H(x_i,x_j)&=\,G(\infty,\ldots,\infty,x_i,\infty,\ldots,\infty,x_j,\infty,\ldots,\infty)= \left[\left( \sqrt[n-1]{  F(x_i)}+\sqrt[n-1]{  F(x_j)}-1\right)^+\right]^{n-1} \,.\label{dependent}
 \end{align}
We also study the independent distribution for which
\begin{gather}
\text{$G(x_1,x_2,\ldots,x_n)=\prod_{i=1}^nF(x_i)$ and $H(x_i,x_j)=F(x_i)F(x_j)$.}\label{ind}
\end{gather}

Using a joint distribution satisfying Clayton's Copula in (\ref{defG}) or
the independence condition in (\ref{ind}) gives the following   specification of the randomized allocation algorithm introduced by Lu and Yu \cite{ly08}.

\medskip
 \hrule
  \begin{algorithm}\label{alg1}
 {\sc Input}: A processing time matrix  $t\in\mathbb R_+^{2\times n}$.
  \\{\sc Output}: A randomized allocation $x\in\{0,1\}^{2\times n}$.
  \begin{enumerate}
    \item Choose random variables $  X_1,  X_2, \ldots, X_n$ according to distribution function $G$
    \vspace{-2mm}\item For each task $j=1,2,\ldots,n$ do
   \vspace{-2mm} \item  \hspace{3mm} if $t_{1j}/t_{2j}<X_j$ then $x_{1j}\leftarrow 1$ else $x_{1j}\leftarrow 0$
    \vspace{-2mm}\item  \hspace{3mm} $x_{2j}\leftarrow 1-x_{1j}$
   \vspace{-2mm} \item End-for
   \vspace{-2mm} \item Output $x$
\end{enumerate}
  \hrule
  \end{algorithm}

\medskip

Let real function $\varphi:\mathbb R_+\times\mathbb R_+\rightarrow\mathbb R$ be defined by
\begin{eqnarray}
\varphi(x,y)=1+y-\min\left\{1,1-\frac1x+y\right\}F(x)-yF(y)+\min\left\{1+\frac1x,1+y\right\}H(x,y)
\label{phi}
\end{eqnarray}

\begin{theorem}\label{th:ratio}
 The approximation ratio of Mechanism \ref{alg1} is at most $\max\{\varphi(x,y):x,y\in\mathbb R_+\}$.
\end{theorem}
\begin{proof}
 For every $j\in [n]$, let $r_j=t_{1j}/t_{2j}$. It has been shown by Lu and Yu \cite{ly08} that the approximation ratio of Mechanism \ref{alg1}  is bounded above by $\max\{\rho_{jk}:j,k\in[n]\} $, where for every pair of distinct indices $i,j\in[n]$,
  \begin{align*}
\rho_{jk} =&\,Pr(x_{1j}=1)+r_k\cdot Pr(x_{1k}=1)+(1/r_j-r_k)^+\cdot Pr(x_{2j}=1,x_{1k}=1)\\
&+(1+1/r_j)\cdot Pr(x_{2j}=1,x_{2k}=1)\,.
\end{align*}
Notice that $X_j\le r_j\Leftrightarrow x_{1j}=0\Leftrightarrow x_{2j}=1$. Hence
\begin{align*}
\rho_{jk} =&\,Pr(X_j> r_j)+r_k\cdot Pr(X_k> r_k)+(1/r_j-r_k)^+\cdot Pr(X_j\le r_j,X_k>r_k)\\
&+(1+1/r_j)\cdot Pr(X_j\le r_j,X_k\le r_k )\\
=&\,\bar F(r_j)+r_k\cdot\bar F(r_k)+(1/r_j-r_k)^+\cdot(F(r_j)-H(r_j,r_k))+(1+1/r_j)\cdot H(r_j,r_k)\\
=&\,1+r_k-\left(1-(1/r_j-r_k)^+\right)F(r_j)-r_kF(r_k)+\left(1+1/r_j-(1/r_j-r_k)^+\right)H(r_j,r_k)\\
=&\,1+r_k-\min\{1,1-1/r_j+r_k\}F(r_j)-r_kF(r_k)+\min\{1+1/r_j,1+r_k\}H(r_j,r_k)
\end{align*}
shows that $\rho_{jk}=\varphi(x_j,x_k)$.
\end{proof}

\section{Strongly independent tasks}\label{indep}
In this section, we consider tasks being allocated strongly independently. Therefore, the joint distribution takes the form $H(x,y)=F(x)F(y)$, giving
\begin{eqnarray}
\varphi(x,y)=1+y-\min\left\{1,1-\frac1x+y\right\}F(x)-yF(y)+\min\left\{1+\frac1x,1+y\right\}F(x)F(y),
\label{independent}
\end{eqnarray}
from which the following symmetry can be proved by elementary mathematics.
\begin{lemma}\label{=}
Let distribution function $G$ satisfy \eqref{ind}. If $F(x)=1-F(1/x)$ for any $x\ge0$, then $\varphi(x,y)=\varphi(1/y,1/x)$ for any $x,y\in\mathbb R_+$.\qed
\end{lemma}

In Section \ref{LuYu}, we point out a mistake {of} \cite{ly08} in estimating over a transcendental function, which invalidates the ratio 1.5963 claimed. In Section \ref{sec-our}, we
introduce
an algebraic piecewise function to construct a class of joint distributions of
independent random variables. Then, we prove that using this class of independent distributions in Algorithm \ref{alg1} gives an improved ratio 1.58606. In Section \ref{sec:limit}, we show the limitation of Algorithm \ref{alg1} for strongly independent tasks, from which { no} ratio better than 1.5852 can be expected.

\subsection{Lu and Yu's transcendental function}\label{LuYu}
  Lu and Yu \cite{ly08} considered function $F(x)=1-\frac1{2^{x^{2.3}}}$. For any $\alpha_1,\alpha_2\in\mathbb R_+$, let $\beta_1=F(\alpha_1)$ and $\beta_2=F(1/\alpha_2)$. By Theorem 4 and in particular the instance on page 410 of  \cite{ly08}, Lu and Yu's mechanism has approximation ratio at least
\begin{align*}
\theta(\alpha_1,\alpha_2)=&\,(1+\alpha_2)\beta_1\beta_2+\beta_1(1-\beta_2)+(1+\alpha_1)(1-\beta_1)(1-\beta_2)+\max\{\alpha_1,\alpha_2\}\beta_2(1-\beta_1)\\
=&\,\left\{\begin{array}{ll}(1+\alpha_2)\beta_1\beta_2+\beta_1(1-\beta_2)+(1+\alpha_1)(1-\beta_1)(1-\beta_2)+ \alpha_2 \beta_2(1-\beta_1),&\text{ if }\alpha_2\ge\alpha_1\,;\vspace{1mm}\\
(1+\alpha_2)\beta_1\beta_2+\beta_1(1-\beta_2)+(1+\alpha_1)(1-\beta_1)(1-\beta_2)+ \alpha_1 \beta_2(1-\beta_1),&\text{ if }\alpha_1\ge\alpha_2\,.\end{array}\right.
\end{align*}
They claimed in Theorem 5 of  \cite{ly08} that under this $F(x)$, $\theta(\alpha_1,\alpha_2)\le1.5963$. However, a contradiction is given by
\begin{align*}
&\theta(0.87793459260323,2.09409917605545
)=1.64065136465694\,.
\end{align*}
  In view of this, the previously best known approximation ratio for truthful scheduling on two unrelated machines was 1.6737 in Lu and Yu's earlier conference paper~\cite{ly2008-stacs}. 
  In this paper, we reduce the ratio to 1.58606 by defining $F$ to be a piecewise algebraic function.

\subsection{An algebraic piecewise function} \label{sec-our}
The challenging task in implementing Mechanism \ref{alg1} is the selection of distribution function $G$. In the case of strongly independent tasks, it amounts to selecting function $F$ such that the maximum of $\varphi$ is as small as possible. To the best of our knowledge, the functions studied in previous work for multiple tasks are either noncontinuous or non-algebraic \cite{ly2008-stacs,ly08,nr01}. In this subsection, we show that the combination of continuity and simple algebraic form beats previous functions, giving improved approximation ratios.

Suppose that  $a\in[1.7,3]$ and $b\in[0.7,1]$ are constants. We study the following {\em continuous} piecewise {\em algebraic} function
\begin{align}\label{ourformula}
F(x)=\left\{\begin{array}{ll}
1,&x\in I_1=[a,+\infty),\vspace{1mm}\\
1-\frac{2(1-b)(a-x)}{a-1},&x\in I_2=[\frac{a+1}2,a),\vspace{1mm}\\
\frac12+\frac{(2b-1)(x-1)}{a-1} ,&x\in I_3=[1,\frac{a+1}2),\vspace{1mm}\\
\frac12-\frac{(2b-1)(1/x-1)}{a-1} ,&x\in I_4=[ \frac2{a+1},1),\vspace{1mm}\\
\frac{2(1-b)(a-1/x)}{a-1},&x\in I_5=[\frac1a, \frac2{a+1}),\vspace{1mm}\\
0,&x\in I_6=[0,\frac1a),
\end{array}
\right.
\end{align}
where the five {\em demarcation   points} $\frac1a$, $\frac2{a+1}$, 1, $\frac{a+1}2$, $a$ divide the domain $[0,+\infty)$ into six {\em intervals} $I_1, I_2,\ldots, I_6$. The function $F(\cdot)$, when plugged into (\ref{ind}), gives an improvement 0.08764 over the previous best ratio of 1.6737 \cite{ly2008-stacs}. Notice that $F(\cdot)$ enjoys the property that  \begin{equation}\label{symmetric}
F(x)+F(1/x)=1\text{ for any }x\ge0.
\end{equation}
An immediate corollary is $F(1)=0.5$.

\begin{theorem} \label{th:independent}
Let $a=1.715$ and $b=0.76$.
Using $F(x)$ in (\ref{ourformula}) 
and $G(x_1,x_2,\ldots,x_n)$ in (\ref{ind}), Mechanism \ref{alg1} achieves approximation ratio $1.58606$.
\end{theorem}

\begin{proof}
By Theorem \ref{th:ratio}, it suffices to show that the maximum of $\varphi(x,y)$   in (\ref{independent}) is no more than {$\rho^*=1.58606$}.
By (\ref{symmetric}) and Lemma \ref{=}, we may assume $xy\ge 1$, for which the function $\varphi(x,y)$  to be maximized takes the form of
\begin{gather}\label{xy>=1}
\varphi(x,y)=1+y- F(x)-yF(y)+\left(1+\frac1x\right)F(x)F(y)\,.
\end{gather}
Note that $\varphi(x,y)$ is continuous in $\mathbb{R}_{+} \times \mathbb{R}_+$. Suppose $x^*,y^*\in\mathbb R_+$ with $x^*y^*\ge1$ attains the maximum, i.e.,
$(x^*,y^*)\in\arg \max_{xy\ge1}\varphi(x,y)$.

We will show that $ \varphi(x^*,y^*)<\rho^*$ by considering the different
 possible domains of the variables $x$ and, $y$ 
 in a case by case basis.
When  $x^*$ or $y^*$ does not belong to the domain associated with a given case,
we say that $(x^*,y^*)$ does not {\em belong to the case}. We will show that, for any case $x\in I_i$, $y\in I_j$ ($1\le i,j\le6 $)  to which  $(x^*,y^*)$ may belong,  $\varphi(x,y)$ is smaller than $\rho^*$ by upper bounding its value
at critical points {(i.e., when the derivatives of
 $\varphi(x,y)$ are equal to zero)} and {that} at demarcation points.

\paragraph{\sc Case 1.} $x\ge a$. It follows  from (\ref{ourformula}) that $F(x)=1$ and from (\ref{xy>=1}) that
$\varphi(x,y) = y+(1+\frac1x-y)F(y)$. {If} $y\le1+\frac1x$ or $y\ge a$, {then} $\varphi(x,y)\le y+(1+\frac1x-y) =1+\frac1x\le1+\frac1a<1.584$.
 If $y>1+\frac1x$ and $y< a$, {then} $y\in(1,a)$. 
 
In case of $y\in[\frac{a+1}2,a)$, since
   $\frac{\partial\varphi}{\partial x}(x,y)=-\frac{2400y-541}{3575x^2} <0$, from KKT condition, we deduce that $(x^*,y^*)$ does not belong to {this case} 
  unless $x^*=a$. When
  $x=a$, {note that $\varphi\left(a,\frac{a+1}2\right)<1.53$, and that}   $\varphi(a,y)$ has a unique critical point $y=\frac{a^2+a+1}{2a}
$ in $(\frac{a+1}2,a)$ with corresponding critical value less than 1.58602. 

In case of $y\in(1,\frac{a+1}2)$, it suffices to consider the case where $x=a$ as
$\frac{\partial\varphi}{\partial x}(x,y)=-\frac{16y-5}{22x^2}<0$. Note that  $\frac{\partial\varphi}{\partial y}(a,y)=\frac{17949}{7546}-\frac{16}{11}y>2.3-\frac{16}{11}(\frac{a+1}2) >0$, which excludes the possibility of $(x^*,y^*)$ belonging to {this case.} 

\paragraph{\sc Case 2.} $y\ge a> x\ge 0$. Note that
$\varphi(x,y)=1+y- F(x)-y +(1+\frac1x)F(x)=1+\frac{F(x)}x$ is a function of single variable $x$. It is easy to check that the derivative of $1+\frac{F(x)}x$ is positive for all $x\in(\frac1a,a)-\{\frac2{a+1},1,\frac{a+1}2\}$. The continuity of $\varphi$ implies $\varphi(x,y)\le \varphi(a,y)=1+\frac1a<1.584$ for all $x\in (\frac1a,a)$.
  When $x\in(0,\frac{1}a]$, it is clear that $\varphi(x,y)=1$.

\bigskip
Cases 1 and 2 above show that {$\varphi(x,y)<\rho^*$} 
when $x $ or $y$ belongs to $I_1$.
For the remaining cases, we have  $x,y< a$. As $xy\ge1$, we have $x,y>\frac1a$ both contained in  $(\cup_{i=2}^4I_i)\cup(I_5-\{\frac1a\})$. We distinguish among Cases 3 -- 6, where Case $i+1$ deals with for $x\in I_i$, $i=2,3,4$ and Case 6 deals with $x\in I _5-\{\frac1a\}$. 

\paragraph{\sc Case 3.} $x\in I_2=[\frac{a+1}2,a)$. We distinguish among four subcases for $y\in[\frac{a+1}2,a)$, $y\in[1,\frac{a+1}2)$, $y\in[\frac2{a+1},1)$, and $y\in(\frac{1}{a},\frac{2}{a+1})$, respectively.

\medskip\noindent{\sc Case 3.1.} $y\in[\frac{a+1}2,a)$. 
In case of $x,y\in(\frac{a+1}2,a)$, solving $\frac{\partial\varphi}{\partial x}(x,y)=0=\frac{\partial\varphi}{\partial y}(x,y)$ gives
$ \frac{987.84x^2+29.2681}{576x^2-129.84}$ $=y=\frac{2400x^2+7990.125x-541}{7150x}$, which implies
{$2764800x^4-4921488x^3+1656341.83x-140486.88=0$.} Among the four real roots of the biquadratic equation, only one root $x_0\doteq1.5419$ belongs  $I_2=[\frac{a+1}2,a)$. 
So   function $\varphi(x,y)$ has a  unique critical point $(x_0,y_0)$, where $y_0=\frac{987.84x_0^2+29.2681}{576x_0^2-129.84}\doteq1.586
$, giving critical  value $\varphi(x_0,y_0)<1.585
$.

In case of $x=\frac{a+2}2$, function $\varphi(\frac{a+1}2,y)$ has a unique critical point $y_0=\frac{a^2+a+ab+3b}{2a+2}\doteq1.5174$ in $ (\frac{a+1}2,a)$,  giving  critical value smaller than $ 1.586059$. Note that $\varphi(\frac{a+1}2,\frac{a+1}2)<1.57$. 

In case of $y=\frac{a+1}2$ and $x\in(\frac{a+1}2,a)$, the derivative of  $\varphi(x,\frac{a+1}2)$ is $\frac{10279}{89375x^2}-\frac{576}{3575}<\frac{10279}{89375 }-\frac{576}{3575}<0$, saying that $(x^*,y^*)$ does not belong to this case.

\medskip\noindent{\sc Case 3.2.} $y\in[1,\frac{a+1}2)$. {Similar arguments to that in Case 3.1 show the following: In case of $x\in(\frac{a+1}2,a)$ and  $y\in(1,\frac{a+1}2)$, function $\varphi$ attains its critical value $\varphi(x_0,y_0)<1.583
$ at  $x_0\doteq1.5249 $, $y_0= \frac{1053x_0^2+43.95625}{624x_0^2+140.66}\doteq1.566
$. In case of $x=\frac{a+1}2$, 
function $\varphi$ attains its critical value $\varphi(\frac{a+1}2,y_0)<1.585
$ at  $y_0=\frac{a^2-3+2ab-2b-2a+4ab^2+12b^2}{4(a+1)(2b-1)}\doteq1.5037
$; at the boundary, $\varphi(\frac{a+1}2,1)<1.4$.
In case of $y=1$ and $x\in(\frac{a+1}2,a)$,   $(x^*,y^*)$ does not belong to this case.}

\medskip\noindent{\sc Case 3.3.} $y\in[\frac2{a+1},1)$.
If $\frac{\partial\varphi}{\partial x}(x,y)=0$, then  $x^2
=\frac{5.41}{24}\left(\frac{70.4}{16-5y}-5.4\right)<\frac{5.41}{24}\left(\frac{70.4}{16-5}-5.4\right)<0.3$, contradicting the hypothesis $x\in[\frac{a+1}2,a)$ of Case 3. Thus $\frac{\partial\varphi}{\partial x}(x,y)\ne 0$, and it suffices to consider the case where $x=\frac{a+1}2$. Note that the derivative of $\varphi(\frac{a+1}2,y)$ is $\frac{143336}{149325y^2}-\frac{5}{22}>\frac{143336}{149325}-\frac{5}{22}>0$. We deduce that  $(x^*,y^*)$ does not belong to Case 3.3.

\medskip\noindent{\sc Case 3.4.} $y\in(\frac1a,\frac2{a+1})$. {It can be shown that $(x^*,y^*)$ does not belong to this case by arguments similar to that in Case 3.3.} 

\paragraph{\sc Case 4.} $x\in I_3=[1,\frac{a+1}2)$. It follows from $xy\ge1$ that $y>\frac2{a+1}$ for which we distinguish among three {subcases} for $y\in[\frac{a+1}2,a)$, $y\in[1,\frac{a+1}2)$ and $y\in[\frac2{a+1},1)$, respectively.

In case of $y\in[\frac{a+1}2,a)$, for $x\in(1,\frac{a+1}2)$ and $y\in(\frac{a+1}2,a)$, {function $\varphi$ attains   critical value $\varphi(x_0,y_0)<1.5854
$ at  the unique critical point $(x_0,y_0)$, where $x_0\doteq1.2027 $ and $y_0=\frac{26754x_0^2 +{35165}/{32}}{15600x_0^2+4875}\doteq1.4504
$.}
Note that $\varphi(1,\frac{a+1}2)=1.5858$. For $x=1$ and $y\in(\frac{a+1}2,a)$, the derivative  of $\varphi(1,y)$ is  negative.
For $y=\frac{a+1}2$ and $x\in(1,\frac{a+1}2)$, the derivative of $\varphi(x,\frac{a+1}2)$ is negative. {It follows that $(x^*,y^*)$   belongs to neither of the two cases.}

In case of $y\in[1,\frac{a+1}2)$,
if $\frac{\partial\varphi}{\partial x}(x,y)=0$, then  $x^2=\frac{6.875}{27-16y}-\frac5{16}<\frac{6.875}{27-8(a+1)}-\frac5{16}
<1
$, {contradicting  the hypothesis $x\ge1$  of Case 4}. So {it suffices to consider} $x=1$. Within $y\in(1,\frac{a+1}2)$, function $\varphi(1,y)$ attains its unique critical value $\varphi(1,\frac{43}{32})<1.586$ at $y=\frac{43}{32}
$. At the boundary, we have $\varphi(1,1)=1.5$.

In case of $y\in(\frac2{a+1},1)$, if $\frac{\partial\varphi}{\partial x}(x,y)=0$, then  $x^2=\frac{80-135y}{80y- 256}$, which along with $x\ge1$ enforces $y\ge\frac{336}{215}$, a contradiction to $y<1$. Therefore we may assume $x=1$. Since the derivative of $\varphi(1,y)$ is $\frac{16-5y^2}{22y^2}>0$, we deduce that $(x^*,y^*)$ does not belong to this case.

\paragraph{\sc Case 5.} $x\in I_4=[\frac{2}{a+1},1)$. It follows from $xy\ge1$ that $y>1$. We distinguish between two subcases depending on whether $y$ is at least $\frac{a+1}2$ or not.

{In case of} $y\in[\frac{a+1}2,a)$, if $\frac{\partial\varphi}{\partial y}(x,y)=0$, then   $y=\frac{2524.47 x^2+429x-624}{1716x^2}
$, which along with $y\ge\frac{a+1}2$ enforces $5x^2+11x-16\ge0$ implying $x\le-3.2$ or $x\ge1$, a contradiction to the hypothesis $x\in I_4$ of Case 5. So we may assume $y=\frac{a+1}2$. Within $x\in(\frac2{a+1},1)$, the unique critical point of $\varphi(x,\frac{a+1}2)$ is $x=\frac{608}{609}$, giving critical value less than 1.586. 
At the boundary, 
we have  \mbox{$\varphi(\frac{2}{a+1},\frac{a+1}2)<1.52
$.}

{In case of $y\in(1,\frac{a+1}2)$, when $x\in(\frac2{a+1},1)$, function $\varphi$ attains its unique critical value $\varphi(x_0,y_0)<1.58603$ at   $x_0\doteq0.985
$, $y_0=\frac{50193x_0/16+1690}{5408-1859x_0}\doteq1.3364
$.
When} $x=\frac2{a+1}$, function $\varphi(\frac2{a+1},y)$ has a unique critical value
$\varphi(\frac2{a+1},1.12665)<1.56
$. 


\paragraph{\sc Case 6.} $x\in I_5-\{\frac1a\}=(\frac1a,\frac{2}{a+1})$. It follows from $xy\ge1$ that $y\in(\frac{a+1}2,a)$,
 If $\frac{\partial\varphi}{\partial x}(x,y)=0$, then   $x=\frac{2400y-541}{637637/400+858y}
$, which along with $x\le\frac2{a+1}$ enforces $y\le  \frac{4657}{4800}<1$, a contradiction.  Since  $\frac{\partial\varphi}{\partial x}(x,y)$ is a continuous function, we deduce that $\frac{\partial\varphi}{\partial x}(x,y)$ is always positive or always negative, implying that $(x^*,y^*)$ does not belong to Case 6.

\bigskip Among all cases analyzed above (see Table \ref{tb1} for a partial summary), the bottleneck $1.58605822203599$ ($<\rho^*$) is attained by Case 3.1 with $\varphi\left(\frac{a+1}2,\frac{a^2+a+ab+3b}{2a+2}\right)=\varphi(1.3575,1.51742633517495)$.
\end{proof}

 \begin{table}[h!]
\begin{center}
\renewcommand{\arraystretch}{1.45}
\begin{tabular}{c| c | c|  c| c}
\hline
Case & Hypothesis & $x_0$  &$y_0$   &$ \varphi(x_0,y_0)$  \\ [0ex] %
 \hline

 {1}& $x\ge a>y\ge\frac{a+1}2 $  & $ a=1.715 $  & $\frac{a^2+a+1}{2a}=1.6490$ & $1.58601068358666$\\[0ex]
\hline

{3}& $x,y\in[\frac{a+1}2,a)$ & $\frac{a+1}2=1.3575$    & $\frac{a^2+a+ab+3b}{2a+2}=1.517426335
$  &1.58605822203599\\[0ex]
\hline

 & & 1.2027121359
    & 1.45036644115936  & 1.58531963915869   \\ [-1.1ex]

 \raisebox{1.3ex}{4} & \raisebox{1.3ex}{$a>y\ge \frac{a+1}2>x\ge1$} &1&$\frac{a+1}2=1.3575$ & 1.5858   \\ [-0ex]
 \hline

 {4}&$x,y\in[1,\frac{a+1}2)$& 1     & 1.34275& 1.5859375 \\[0ex]
\hline

{5}&$a>y\ge\frac{a+1}2$, $1>x\ge\frac2{a+1}$& 0.9983579639     & $\frac{a+1}2=1.3575$  &1.58580149521531 \\[0ex]
\hline

{5}& $\frac{a+1}2> y \ge1>x\ge\frac2{a+1}$  &0.98503501986    &  1.33641518393347  & 1.58602337235828
  \\ [0ex] 
\hline 
\end{tabular}
\caption{The cases in the proof of Theorem \ref{th:ratio} where $\varphi(x,y)$ exceeds 1.585.}
\label{tb1}
\end{center}
\end{table}

 \subsection{The limitation of Mechanism \ref{alg1}}\label{sec:limit}
It was announced in \cite{ly08} and proved in its full paper that, for strongly independent tasks, that the performance ratio of Mechanism \ref{alg1} cannot be better than 1.5788. We improve the lower bound by 0.0074, which nearly closes the gap between the lower and upper bounds for Mechanism \ref{alg1}.

\begin{theorem}\label{th:limit}
Let distribution function  $G$ in (\ref{ind}) be defined by  any non-decreasing function $F:\mathbb R_+\rightarrow [0,1]$ with $F(0)=0$ and $\lim_{x\rightarrow\infty}F(x)=1$. The approximation ratio of Mechanism \ref{alg1}    is at least $1.5852$.
\end{theorem}
\begin{proof}
Suppose that there exists function $F$ such that Mechanism \ref{alg1} achieves an approximation ratio   less than 1.5852. It follows from (\ref{independent}) that for any $x,y\in\mathbb R_+$,
\begin{eqnarray}\label{limit}
1.5852>\varphi(x,y)=\left\{
\begin{array}{ll}
1+y-F(x)-yF(y) +\left(1+\frac1x\right)F(x)F(y),&xy\ge1;\vspace{1mm}\\
1+y- \left(1-\frac1x+y\right)F(x)-yF(y)+(1+y)F(x)F(y),&xy\le1.
\end{array}
\right.
\end{eqnarray}
Let $\alpha=1.352$ and $\beta=1.532$. We examine $\varphi(x,y)$ for some values of $x,y$ in $\{\alpha,\beta,1, 1/{\alpha},1/{\beta}\}$, and derive a contradiction to $\varphi(x,y)<1.5852$.

First, we investigate several values $x,y\ge1$ to which the first row of (\ref{limit}) applies. It follows from
$1.5852>\varphi(\alpha,\alpha)=\frac{294}{169}(F(\alpha))^2-\frac{294}{125}F(\alpha)+\frac{294}{125} $ that
 \begin{gather}\label{r2}
0.54<F(\alpha)< \frac{169}{250}+\frac{13 \sqrt{1174}}{3500}<0.81\,.
\end{gather}
It follows from $1.5852>\varphi(\beta,\beta)=\frac{633}{383}(F(\beta))^2 -\frac{633}{250}F(\beta)+\frac{633}{250}$ that
\begin{gather}\label{r4}
F(\beta)<\frac{383}{500}+\frac{\sqrt{154595269}}{105500}=\lambda_1\,.
\end{gather}
It follows from  $1.5852>\varphi(\alpha ,\beta )=\frac{294}{169}\cdot F(\alpha )\cdot F(\beta ) -F(\alpha )-\frac{383}{250} F(\beta )+\frac{633}{250}$ that
  $\left(\frac{383}{250} -\frac{294}{169}\cdot F(\alpha )\right)F(\beta )>\frac{2367}{2500}-F(\alpha )$.
 Notice from $F(\alpha)<0.81$ in (\ref{r2}) that  $\frac{383}{250} -\frac{294}{169}\cdot F(\alpha )>0$. Therefore (\ref{r4}) implies $\lambda_1>F(\beta )>\frac{2367/2500-F(\alpha )}{383/250 -294\cdot F(\alpha )/169}$, giving
\begin{eqnarray}
  F(\alpha )<\frac{169(3830\lambda_1-2367)}{2500(294\lambda_1-169)} =\lambda_2
  \,.\label{r5}
\end{eqnarray}
It follows from $1.5852>\varphi(1,\alpha )=2  F(1)\cdot F(\alpha )- F(1)-\frac{169}{125}   F(\alpha )+\frac{294}{125}$ that $(2 F(\alpha )-1)F(1)<\frac{169}{125}   F(\alpha )-\frac{1917}{2500}$. Since $F(\alpha )>0.5$ by (\ref{r2}), we have
$F(1)<\frac{169 \cdot F(\alpha )/125-1917/2500}{2\cdot F(\alpha )-1}=\frac{169}{250}-\frac{227/2500}{2\cdot F(\alpha )-1}$, which along with   (\ref{r5}) gives
\begin{gather}\label{r6}
F(1)<\frac{169}{250}-\frac{227/2500}{2\lambda_2-1}=\lambda_3
\,.
\end{gather}

Next, we examine values for $x,y\le1$. Using (\ref{limit}), we obtain $1.5852>\varphi(\frac1{\alpha },1)=2-\frac{81}{125}\cdot F(\frac1{\alpha })-F(1)+2\cdot F(1)\cdot F(\frac1{\alpha })$, i.e.,
\[(2\cdot F(1)-0.648)F\left(\frac1{\alpha }\right)<F(1)-0.4148\,.\]
If $F(1)\le 0.324$, then $F(1)< 0.324$ and $F\left(\frac1{\alpha }\right)>\frac{F(1)-0.4148}{2F(1)-0.648}=\frac12+\frac{0.0908}{0.648-2F(1)}>0.5$, giving a contradiction to $F\left(\frac1{\alpha }\right)\le F(1)<0.324$. Hence $F(1)> 0.324$ and $F\left(\frac1{\alpha }\right)
<\frac{F(1)-0.4148}{2F(1)-0.648}=\frac12-\frac{0.0908}{2F(1)-0.648}$. By (\ref{r6}) we have
\begin{eqnarray}\label{r7}
F\left(\frac1{\alpha }\right)<\frac12-\frac{0.0908}{2\lambda_3-0.648}=\lambda_4
\,.
\end{eqnarray}
        From $1.5852>\varphi (\frac1{\beta },\frac1{\alpha } )=\frac{294}{169}F(\frac1{\beta } )F(\frac1{\alpha }) -\frac{8773}{42250}F (\frac1{\beta } )-\frac{125}{169}F (\frac1{\alpha } )
+
\frac{294}{169}$, we deduce that
\begin{equation}\label{eq}
\left(73500\cdot F\left(\frac1{\alpha}\right)-8773\right)F\left(\frac1{\beta}\right)<31250\cdot F\left(\frac1{\alpha}\right)-6525.3\,.\end{equation}
 If $F(\frac1{\alpha})\le\frac{8773}{73500}$, then $31250\cdot F (\frac1{\alpha} )-6525.3<0$, implying $F(\frac1{\alpha})<\frac{8773}{73500}<0.2$ and
  \[F \left(\frac1{\beta } \right)>\frac{31250\cdot
F \left(\frac1{\alpha } \right)-6525.3}{{73500\cdot F(\frac1{\alpha })- {8773} }}
=\frac{125}{294}+\frac{{6525.3-\frac{125}{294}\cdot {8773} }}{{8733-735000\cdot F (\frac1{\alpha } ) }}>\frac{125}{294}>0.4>F\left(\frac1{\alpha}\right).\] However $F(\frac1{\beta})>F(\frac{1}{\alpha})$ contradicts the fact that $F$ is non-decreasing. Thus $F(\frac1{\alpha})>\frac{8773}{73500}$, and
it follows from (\ref{eq}) that $F (\frac1{\beta } )<\frac{125}{294}-\frac{6525.3-1096625/294}{73500\cdot F(\frac1{\alpha})-8773}$. In turn (\ref{r7}) implies 
\begin{gather}\label{final}
F\left(\frac1{\beta }\right)<\frac{125}{294}-\frac{6525.3-1096625/294}{73500\lambda_4-8773}
<0.1143\,.
\end{gather}
On the other hand, we deduce from $1.5852>\varphi (\frac{1}{\beta },\frac{1}{\beta } )=\frac{633}{383} (F (\frac1{\beta } ) )^2-\frac{74061}{95750}
F (\frac1{\beta } )+\frac{633}{383}$ that $F (\frac1{\beta } )> \frac{117}{500}-\frac{\sqrt{154595269}}{105500}>0.116$, contradicting (\ref{final}).
\end{proof}

In the previous proof of lower bound $1.5788$ \cite{ly08}, Lu and Yu showed that for a parameter $\gamma\doteq1.434$, the values $\varphi(\gamma,1/\gamma)$ and $\varphi(\gamma,1)$ cannot be both smaller than 1.5788 no matter what $F$ is chosen. As seen from the above, our improved lower bound 1.5852 is established by introducing two parameters $\alpha=1.352$, $\beta=1.532$, and considering function value $\varphi$ at seven point: $(\alpha,\alpha),(\beta,\beta),(\alpha,\beta),(1,\alpha),(1/\alpha,1),(1/\beta,1/\alpha)$ and $(1/\beta,1/\beta)$.
\section{Weakly independent tasks}\label{dep}
 We assume function $F(\cdot)$ takes the form of (\ref{ourformula}). The weak independence is specified by the joint distribution $H(x_i,x_j)= \left[\left( \sqrt[n-1]{  F(x_i)}+\sqrt[n-1]{  F(x_j)}-1\right)^+\right]^{n-1}$ as in (\ref{dependent}).

Using the Copula based distribution, Mechanism \ref{alg1} can guarantee approximation 1.5067711  for $n=2$ tasks, as proved in Section \ref{sec:n=2}. We study the case of $n\ge3$ tasks in Section \ref{sec:nge3}, where MATLAB's global solver is used to solve the optimization problems involved in the computer conducted search/proof of the approximation ratio. Our results show that the Clayton Copula based algorithm outperforms the strong independent-task allocation, and the former converges to the later as $n$ approaches to infinity.

\subsection{The case $n=2$} \label{sec:n=2}
In this subsection, we reduce Lu's upper bound $\frac16(\sqrt{25-12\sqrt3}+7)\doteq 1.5089$ \cite{l09} for two tasks by $0.00208$, which narrows the gap from the lower bound $1.506$ \cite{l09} to be 0.0007711.
For the case of $n=2$, we have $H(x_1,x_2)=  \left(   F(x_1)+ {  F(x_2)}-1\right)^+ $ and $\varphi(x,y)=1+y-\min\left\{1,1-\frac1x+y\right\}F(x)-yF(y)+\min\left\{1+\frac1x,1+y\right\}(F(x)+F(y)-1)^+$.

\begin{lemma}\label{simplify}
Let distribution function $G$ satisfy \eqref{defG}.
When $n=2$, $\varphi(x,y)=\varphi(1/y,1/x)$ for any $x,y\in\mathbb R_+$.
\end{lemma}

\begin{proof} Without loss of generality we may assume $xy\ge1$. Since $F(\cdot)$ is non-decreasing and satisfies (\ref{symmetric}), we have $F(x)\ge F(1/y)=1-F(y)$, and $\varphi(x,y)=1+y-F(x)-yF(y)+(1+1/x)(F(x)+F(y)-1)$

On the other hand, $F(1/y)+F(1/x)\le F(1/y)+F(y)=1$ implies $\varphi(1/y,1/x)=1+1/x-(1-y+1/x)F(1/y)-F(1/x)/x=1+1/x-(1-y+1/x)(1-F(y))-(1-F(x))/x$. Now it is easy to check that $\varphi(x,y)=\varphi(1/y,1/x)$.
\end{proof}

Lu's approximation ratio $  1.5089$ \cite{l09} was proved by choosing $F$ to be a continuous algebraic function piecewise-defined on four intervals according to a constant parameter. Next, we show that our piecewise algebraic function in (\ref{ourformula}), with appropriate choices of {\em two} constants $a$ and $b$, provides an improved approximation ratio.}
\begin{theorem}\label{th:n=2}
Let $F(\cdot)$ be defined as in (\ref{ourformula}) with $a=2.2468$ and $b=0.7607$.
For $n=2$, using
  $G(x_1,x_2)$ in (\ref{defG}), Mechanism \ref{alg1} achieves approximation ratio $1.5067711$.
\end{theorem}
\begin{proof} By the setting of $a$ and $b$, we see that $a-1$, $1-b$, $2b-1$,  $3b-2$, $2ab-a-1$, $a+1-4b$, $2a-3ab+b$, $3a+1-4ab$ are all positive. We will use this fact implicitly in our analysis.
In view of Lemma \ref{simplify}, it suffices to consider $xy\ge1$, $F(x)\ge F(1/y)=1-F(y)$, and
\begin{eqnarray}
\varphi(x,y)&=&1+y- F(x)-yF(y)+ \left(1+\frac1x\right)(F(x)+F(y)-1) \nonumber\\
&=&y-\frac1x+\frac1xF(x)+\left(1+\frac1x-y\right)F(y)\label{n=2}
\end{eqnarray}
In the following, we will consider $x\ge a$ in Case 1, $y\ge a$ in Case 2, and $\max\{x,y\}<a$ in Cases 3 -- 6. Let $(x^*,y^*)$ with $x^*,y^*>0$ and $x^*y^*\ge1$ maximize $\varphi(x,y)$ in (\ref{n=2}).

\paragraph{\sc{Case 1. $x\ge a$.} } It follows from $F(x)=1$ that $\varphi(x,y)=y+(1+\frac1x-y)F(y)$. In case of $y\le1+\frac1x$ or $y\ge a$, we have \begin{center}$\varphi(x,y)\le y+\left(1+\frac1x-y\right) =1+\frac1x\le1+\frac1a<1.5$.\end{center}
  In case of $y>1+\frac1x$ and $y< a$, we have $y\in(1,a)$. 

When $y\in[\frac{a+1}2,a)$, it  
follows from (\ref{ourformula}) and (\ref{n=2}) that
\begin{center}$\frac{\partial\varphi}{\partial x}(x,y)=-\frac{2ab-a-1+2(1-b)y}{(a-1)x^2}
<0$.
  \end{center}
  By KKT condition, we see that $(x^*,y^*)$ does not belong to this  case unless $x^*=a$. In case of $x=a$, function $\varphi(a,y)$ has a unique critical point $y=\frac{a^2+a+1}{2a}$ in $(\frac{a+1}2,a)$, at which the critical value is less than 1.50677. At the boundary point
, we have $\varphi\left(a,\frac{a+1}2\right)<1.5$.

When $y\in(1,\frac{a+1}2)$, since 
$\frac{\partial\varphi}{\partial x}(x,y)=-\frac{a+1-4b+(4b-2)y}{2(a-1)x^2}
<0$, it suffices to consider the case where $x=a$. Note that  $\frac{\partial\varphi}{\partial y}(a,y)=\frac{17770909}{11672126}-\frac{869}{1039}y>\frac{3}{2}-\frac{869}{1039}(\frac{a+1}2) >0.1$, which excludes the possibility of $(x^*,y^*)$ belonging to this case.

\paragraph{\sc Case 2.} $y\ge a> x\ge 0$. It follows from $F(y)=1$ that $\varphi(x,y)=1+\frac{F(x)}x$ is a function of single variable $x$. When $x\in[0,\frac{1}a]$, it is clear that $\varphi(x,y)=1$. The derivative of $\varphi(x,y)$ is $\frac{2(1-b)(2-ax)}{(a-1)x^3}>0$ for $x\in(\frac1a,\frac2{a+1})$, $-\frac{a+1-4b}{(a-1)x^2}<0$ for $x\in (1,\frac{a+1}2)$, and $-\frac{2ab-a-1}{(a-1)x^2}<0$ for $x\in (\frac{a+1}2,a)$.  So we may assume $x\in[\frac2{a+1},1]$
Within $x\in(\frac2{a+1},1)$, the derivative of  $\varphi(x,y)
$ has  a unique root $x_0=\frac{4(2b-1)}{a-3+4b}$. So we only need to consider $\varphi(x_0,y)$, $\varphi(1,y)$, and   $\varphi(\frac2{a+1},y)=\frac{3+a-b-ab}2
 $. All three values are smaller than $1.505$.

\bigskip
In the following case analysis, we consider only $x,y< a$. As $xy\ge1$, we have $x,y>\frac1a$. We distinguish among four cases for $x\in[\frac{a+1}2,a)$, $x\in[1,\frac{a+1}2)$, $x\in[\frac2{a+1},1)$, and $x\in(\frac1a,\frac2{a+1})$, respectively.

\paragraph{\sc Case 3.} $x\in[\frac{a+1}2,a)$. We distinguish among four subcases for $y\in[\frac{a+1}2,a)$, $y\in[1,\frac{a+1}2)$, $y\in[\frac2{a+1},1)$, and $y\in(\frac{1}{a},\frac{2}{a+1})$, respectively.

\paragraph{\sc Case 3.1.} $y\in[\frac{a+1}2,a)$.
It follows from (\ref{ourformula}) and (\ref{n=2}) that 
\begin{center}
$\frac{\partial\varphi}{\partial x}(x,y)=\frac{2(b-1)y-4ab+3a+1}{(a-1)x^2}$ and $ \frac{\partial\varphi}{\partial y}(x,y)=\frac{2(1-b)(ax+x-2yx+1)}{(a-1)x}$.
\end{center}
In case of $x,y\in[\frac{a+1}2,a)$, the unique solution of $\frac{\partial\varphi}{\partial x}(x,y)=0=\frac{\partial\varphi}{\partial y}(x,y)$,   $x_0= \frac{1-b}{2a-3ab+b}$ and $y_0=\frac{3a-4ab+1}{2(1-b)}$, gives the critical value $\varphi (x_0,y_0)<1.5061$.

In case of $x=\frac{a+1}2$ and $y\in(\frac{a+1}2,a)$,  function $\varphi(\frac{a+1}2,y)$ has a unique critical point $y_0=\frac{ a^2+2a+3}{2(a+1)}$, giving  critical value $\varphi\left(\frac{a+1}2,y_0\right)<1.5067711$ \footnote{A more accurate upper bound is 1.506771096398094922363952719025.}.

In case of $y=\frac{a+1}2$, function  $\varphi(x,\frac{a+1}2)$ has positive derivative $\frac{ 2a-3ab+b}{(a-1)x^2}
>0
$ for all $x\in(\frac{a+1}2,a)$, saying that $(x^*,y^*)$ does not belong to  this case.


\paragraph{\sc Case 3.2.} $y\in[1,\frac{a+1}2)$.\;\; 
In case of $1<y<\frac{a+1}2<x<a$, solving $\frac{\partial\varphi}{\partial x}(x,y)=0=\frac{\partial\varphi}{\partial y}(x,y)$, we obtain a unique critical point $x_0=\frac{2(2b-1)}{5a+3-8ab}$, $y_0=\frac{3a-4ab+4b-1}{2(2b-1)}$ of $\varphi$, giving critical value $\varphi(x_0,y_0)<1.504$.

In case of $1<y<\frac{a+1}2=x$,  function $\varphi(\frac{a+1}2,y)$ attains its critical value $\varphi(\frac{a+1}2,y_0)<1.504
$ at its unique critical point $y_0=\frac{a^2+8ab-4a+16b-9}{
4(2b-1)(a+1)}
$.

In case of $y=1$ and $x\in(\frac{a+1}2,a)$, the derivative of  $\varphi(x,1)$ is $\frac{3a+1-4ab}{2(a-1)x^2}
>0$, implying that $(x^*,y^*)$ does not belong to  this case.


\paragraph{\sc Case 3.3.} $y\in[\frac2{a+1},1)$.  It follows from (\ref{ourformula}) and (\ref{n=2}) that 
\begin{center}
$\frac{\partial\varphi}{\partial x}(x,y)=\frac{2(1-b)yx^2+(a+1-4b)y +4b-2}{2(a-1)yx^2}$.
\end{center}
If $\frac{\partial\varphi}{\partial x}(x,y)=0$, then  $x^2=-\frac{(a+1-4b)y +4b-2}{2( 1-b)y}
<0$ shows a contradiction. Thus $\frac{\partial\varphi}{\partial x}(x,y)\ne 0$, and it suffices to consider the case where $x=\frac{a+1}2$. Note that the derivative of $\varphi(\frac{a+1}2,y)$ is $\frac{(a+1)(a+1-4b)}{2(a^2-1)}+\frac{(a+3)(2b-1)}{(a^2-1)y^2}>0$. 
We deduce that  $(x^*,y^*)$ does not belong to Case 3.3.

\paragraph{\sc Case 3.4.} $y\in(\frac1a,\frac2{a+1})$.
Since $ \frac{\partial\varphi}{\partial x}(x,y)=\frac{2(1-b)}{(a-1)x^2y}>0$, we deduce that  $(x^*,y^*)$ does not belong to this case.

\paragraph{\sc Case 4.} $x\in[1,\frac{a+1}2)$. It follows from $xy\ge1$ that $y>\frac2{a+1}$, for which we distinguish between two subcases for $y\in[\frac{a+1}2,a)$ and  $y\in[\frac2{a+1},\frac{a+1}2)$, respectively.

{Consider the subcase of $y\in[\frac{a+1}2,a)$.} When $1<x<\frac{a+1}2<y<a$, solving   $\frac{\partial\varphi}{\partial x}(x,y)=0=\frac{\partial\varphi}{\partial y}(x,y) $ gives the unique critical point $x_0=\frac{2(1-b)}{a-2ab+6b-3}$,
$y_0=\frac{3a-4ab+4b-1}{4(1-b)}$, and the corresponding critical value $\varphi(x_0,y_0)<1.503$.
When $x=1$, the unique critical point  of $\varphi(1,y)$ is $y
=\frac{a+2}2$, giving critical value $\varphi(1,\frac{a+2}2)<1.506$.
When $y=\frac{a+1}2$, the derivative of $\varphi(x,\frac{a+1}2)$ is $\frac{(2b-1)(3-a)}{2(a-1)x^2}>0$ for all $x\in(1,\frac{a+1}2)$, 
excluding the possibility of $(x^*,y^*)$ belonging to this case.

{Consider the subcase of $y\in[\frac2{a+1},\frac{a+1}2)$.} 
Note that  $ \frac{\partial\varphi}{\partial x}(x,y)=\frac{(2b-1)(2-y)}{(a-1)x^2}>0$ for all $y\in(1,\frac{a+1}2)$, and   
  $ \frac{\partial\varphi}{\partial x}(x,y)=\frac{2b-1}{(a-1)yx^2}>0$ for all $y\in(\frac2{a+1},1)$.  We deduce that $(x^*,y^*)$ does not belong to this case.

\paragraph{\sc Case 5.} $x\in[\frac{2}{a+1},1)$. It follows from $xy\ge1$ that $y>1$. We distinguish between two subcases depending on whether $y$ is at least $\frac{a+1}2$ or not.

{Consider the subcase of $y\in[\frac{a+1}2,a)$.} When $\frac{2}{a+1}<x<\frac{a+1}2<y$, solving   $\frac{\partial\varphi}{\partial x}(x,y)=0=\frac{\partial\varphi}{\partial y}(x,y)$, we obtain a unique critical point $x_0=\frac{2(5b-3)}{2ab+2b-a-1}$,
$y_0=\frac{ (a+1) (12b-7)}{4(5b-3)}$ corresponding critical value $\varphi(x_0,y_0)<1.50677$.
 When $y=\frac{a+1}2$. the derivative of $\varphi(x,\frac{a+1}2)$ is $ \frac{8b-4-(a-1)(2b-1)x}{2(a-1)x^3}> \frac{8b-4-(a-1)(2b-1)}{2(a-1)x^3} =\frac{(3-a)(2b-1)}{2(a-1)x^3} >0
 $ for all $x\in(\frac{2}{a+1},1)$, saying that $(x^*,y^*)$ does not belong to this case.
  When $x=\frac2{a+1}$, the derivative of $\varphi(\frac2{a+1},y)=\frac{3a-4y+3}{(1-b)(a-1)}>  \frac{3a-4a+3}{(1-b)(a-1)}>0$ for all $y\in(\frac{a+1}2,a)$,    saying that $(x^*,y^*)$ does not belong to this case.

{Consider the subcase of $y\in(1,\frac{a+1}2)$.} When $x<\frac2{a+1}$, solving $\frac{\partial\varphi}{\partial x}(x,y)=0=\frac{\partial\varphi}{\partial y}(x,y)$, we obtain a unique critical point $x_0=\frac{6(2b-1)}{a+8b-5}$,
$y_0=\frac{a+8b-5}{ 3 (2 b-1)}$, corresponding critical value $\varphi(x_0,y_0)<1.506$.
When $x=\frac2{a+1}$, the derivative of function $\varphi(\frac2{a+1},y)$ is $\frac{ab+5b-3-2(2b-1)y}{a-1}>\frac{ab+5b-3-(2b-1)(a+1)}{a-1}=\frac{a(1-b)+3b-2}{a-1}>0$,  saying that $(x^*,y^*)$ does not belong to this case.

\paragraph{\sc Case 6.} $x\in(\frac1a,\frac{2}{a+1})$. It follows from $xy\ge1$ that $y>\frac{a+1}2$.
 If $\frac{\partial\varphi}{\partial x}(x,y)=0$, then it can be deduced   that  $x=\frac2y
$, which along with $x\le\frac2{a+1}$ enforces $y\ge a+1 $, a contradiction to $y<a$.  Thus $\frac{\partial\varphi}{\partial x}(x,y) $ is always positive or always negative, saying that $(x^*,y^*)$ does not belong to Case 6.
\end{proof}

\subsection{The cases $n\ge3$}\label{sec:nge3}
In this subsection, we mainly discuss the multiple task case $n\ge3$. We look for a distribution function $F(\cdot)$ of form (\ref{ourformula}) which minimizes the maximum of the binary function
 \begin{equation}\label{defphi}
\varphi(x,y)\!=\!1 + y - \min\!\left\{\!1,1\!-\!\frac1x\!+\!y\!\right \}F(x) - yF(y) + \min\!\left\{\!1\!+\!\frac1x,1\!+\!y\!\right\}\!
\left[(\sqrt[n-1]{F(x)}\!+\!\!\!\sqrt[n-1]{F(y)}\!-\!1)^+\right]^{n-1}.\end{equation}
To accomplish the task, we need determine the maximum of $\varphi$ for any given constants $a$ and $b$. Theoretically, this can be done in a way similar to the proofs of Theorems \ref{th:independent} and \ref{th:n=2}. In practice, computer-assisted arguments turn out more suitable, as explained below.
\begin{itemize}
\item The above case analyses are simplified by the property that $\varphi(x,y)=\varphi(1/y,1/x)$ (see Lemmas \ref{=} and \ref{simplify}), which allows us to only focus on the case of $xy\ge1$. For $n\ge3$, this property is generally lost due to the {\em complicated term} $[(\sqrt[n-1]{F(x)} + \sqrt[n-1]{F(y)} - 1)^+]^{n-1}$ in (\ref{defphi}). As a result, it might be much more tedious to discuss all possible combinations for  $x,y$  from six intervals $[0,\frac1a]$, $[\frac1a,\frac2{a+1}]$, ..., $[a,+\infty)$  where $F(\cdot)$ is described by different linear expressions.
\item Finding the critical points of $\varphi(x,y)$ becomes more and more challenging as $n$ increases. One has to resort to software for solving equations of high degrees resulting from the complicated term.
\end{itemize}

 We conduct  a case analyses using MATLAB's global optimization tool {\sc GlobalSearch} (cf., \cite{Ugray07})
to help us to solve the nonlinear program $\max_{x,y} \varphi(x,y)$ subject to four constrains $xy\le$ (or $\ge$) 1, $\sqrt[n-1]{F(x)} + \sqrt[n-1]{F(y)}\le$ (or $\ge$) 1, $l_1\le x\le u_1$, $l_2\le y\le u_2$
for different choices of $n$, $a$, and $b$, where $l_1,u_1,l_2,u_2$  specify the intervals containing $x$ and $y$.
The computational results are summarized in Table \ref{tb2}. (More accurate data are presented in Appendix \ref{apx:data}.)
For each input triplet of $n, a, b$,
Table \ref{tb2} provides the values of
$(x^*,y^*)$ which attain the largest value of $\varphi(x,y)$
after
{\sc GlobalSearch}
is employed to solve the nonlinear program 10 times. 
The difference $\delta$ between the largest value of $\varphi(x,y)$ and the smallest one among the 10 computations is also recorded. From the last column of Table \ref{tb2} we observe that $\delta$ does not exceed $1.4066\times10^{-7}$, showing the stability of the computational results.

As the second line (when $n=2$) in Table~\ref{tb2} illustrates, {\sc GlobalSearch}
finds the optimal solution {established} in Theorem~\ref{th:n=2} within numerical
tolerance.
Actually the step of MATLAB program in which the overall maximum is found terminates at the critical point of $\varphi(\frac{a+1}2,y)$ with the message ``Magnitude of directional derivative in search
 direction less than 2*options.TolFun and maximum constraint violation
  is less than options.TolCon.'' 

 \begin{table}[h!]
\begin{center}
\renewcommand{\arraystretch}{1.45}
\begin{tabular}{c| c | c|  c| c|c|c}
\hline
$n$ & $a$ & $b$ & $x^*$  &$y^*$   &$ \varphi(x^*,y^*)$ &$\delta$ \\ [0ex] %
 \hline

2& 2.2468 & 0.7607 & $ \frac{a+1}2=1.6234 $  & 1.9313955486
& 1.5067710964
&$1.5499\times10^{-13}$\\[0ex]
\hline

3& 1.9328& 0.7418 & 1.9105670668
    & 1.7231009560
    &1.5412707361
 &$5.4073\times10^{-9}$   \\[0ex]
\hline

 4&1.8442 & 0.7453
    & $a=1.8442$  & 1.6932202823
    &1.5559952305
     &$8.8818\times10^{-16}$\\ [0ex]
\hline

5&$1.8070$& 0.7487     & 1.1418758036
& 1.5193285944
 &1.5634859375
 &$1.8911\times10^{-9}$\\[0ex]
\hline

6&$ 1.7863$& $0.7510 $&1.1468400067
    &1.4989121029
      &1.5679473463
 &$3.4101\times10^{-9}$     \\[0ex]
\hline

7& $1.7734$  &0.7526    &  1.1447309125
  & 1.4845715829
 &1.5709131851
&$2.7397\times10^{-8}$ \\ [0ex]
\hline

8&1.7646 &0.7536 &1.1192295299
&1.4661575387
&1.5730320737
&$4.9022\times10^{-9}$\\ [0ex]
\hline

9& 1.7581&0.7543&$\frac{a+1}2=1.37905$ &1.5499380481
&1.5746303803
&$2.1302\times10^{-9}$\\ [0ex]
\hline

10& 1.7530 &0.7548 &1.0673757071
&1.4334673997
&1.5758769995
&$4.5725\times10^{-8}$\\ [0ex]
\hline

15& 1.7410 &0.7570 &1.0190835924
&1.3975512392
&1.5795353027
&$3.2335\times10^{-8}$\\ [0ex]
\hline

20&1.7326&0.7573&0.9997077878
&1.3798783532
&1.5811826690
&$8.8565\times10^{-10}$\\ [0ex]
\hline

30&1.7267&0.7582&$\frac{a+1}2=1.36335$
&1.5259350403
&1.5828322598
&$2.3226\times10^{-13}$\\ [0ex]
\hline

45&1.7225&0.7587&0.9879452462
&1.3491108561
&1.5839252561
&$1.4493\times10^{-9}$
\\ [0ex]
\hline

70&1.7199&0.7592&0.9868820343
&1.3445069231
&1.5846893837
&$3.2863\times10^{-9}$
\\ [0ex]
\hline

100&1.7183&0.7594&$\frac{a+1}2=1.35915$
&1.5197905945
&1.5850948285
&$3.1020\times10^{-13}$
\\ [0ex]
\hline

200&1.7167&0.7597&$\frac{a+1}2=1.35835$
&1.5186228330
&1.5855735653
&$7.5118\times10^{-13}$\\ [0ex]
\hline

500&1.7156&0.7598&0.9851752572
&1.3375636313
&1.5858603200
&$1.8349\times10^{-9}$\\ [0ex]
\hline

1000&1.7153&0.7599&$\frac{a+1}2=1.35765$
&1.5176140596
&1.5859488980
&$3.2567\times10^{-12}$\\ [0ex]
\hline

5000&1.7150&0.7599&0.9849521898
&1.3365770913
&1.5860275919
&$2.9110\times10^{-9}$\\ [0ex]
\hline

$10^4$&1.7149&0.7599& $a=1.7149$
&1.6490128248
&1.5860403769
&$1.4479\times10^{-11}$\\ [0ex]
\hline

$10^5$&1.7149&0.7599&0.9849621198
&1.3364590898
&1.5860442151
&$5.2509\times10^{-9}$\\ [0ex]
\hline

$10^6$&1.7149&0.7599&0.9849513401
&1.3364514130
&1.5860456086
&$1.0466\times10^{-7}$\\ [0ex]
\hline

$\vdots$&$\vdots$&$\vdots$&$\vdots$
&$\vdots$
&$\vdots$\\ [0ex]
\hline

$\infty$&1.715&0.76& $\frac{a+1}2=1.3575$    & $\frac{a^2+a+ab+3b}{2a+2}\doteq1.5174 
$  &1.5860582220
&0\\ [0ex]
\hline
\end{tabular}
\caption{Computational results on minimizing the maximum of $\varphi$ (choosing $a$ and $b$ to minimize the maximum $\varphi(x^*,y^*)$), where the data in the last row for $n=\infty$ are taken from the proof of Theorem \ref{th:independent}.} 
\label{tb2}
\end{center}
\end{table}

The   second to last column of Table \ref{tb2} shows that $\varphi(x^*,y^*)$ increases as $n$ grows, interpreting the common sense that achieving truthfulness with respect to more tasks costs more. The increasing property of approximation ratios with respect to $n$ is illustrated in Figure \ref{fg:increase}, where we have the following observations.
\begin{itemize}
\item The curve makes a ``large'' jump at $n=3$, from 1.5068 to 1.5413;
\vspace{-2mm}\item The increasing speed is tiny after $n=30$, which attains $\varphi(x^*,y^*)\doteq1.5828$;
\vspace{-2mm}\item The curve looks flat after $n=100$; in particular the average slope is less than $5\times10^{-6}$ for $n\in[100,200]$.
\end{itemize}

\begin{figure}[htpb]
\begin{center}
\includegraphics[scale=0.45]{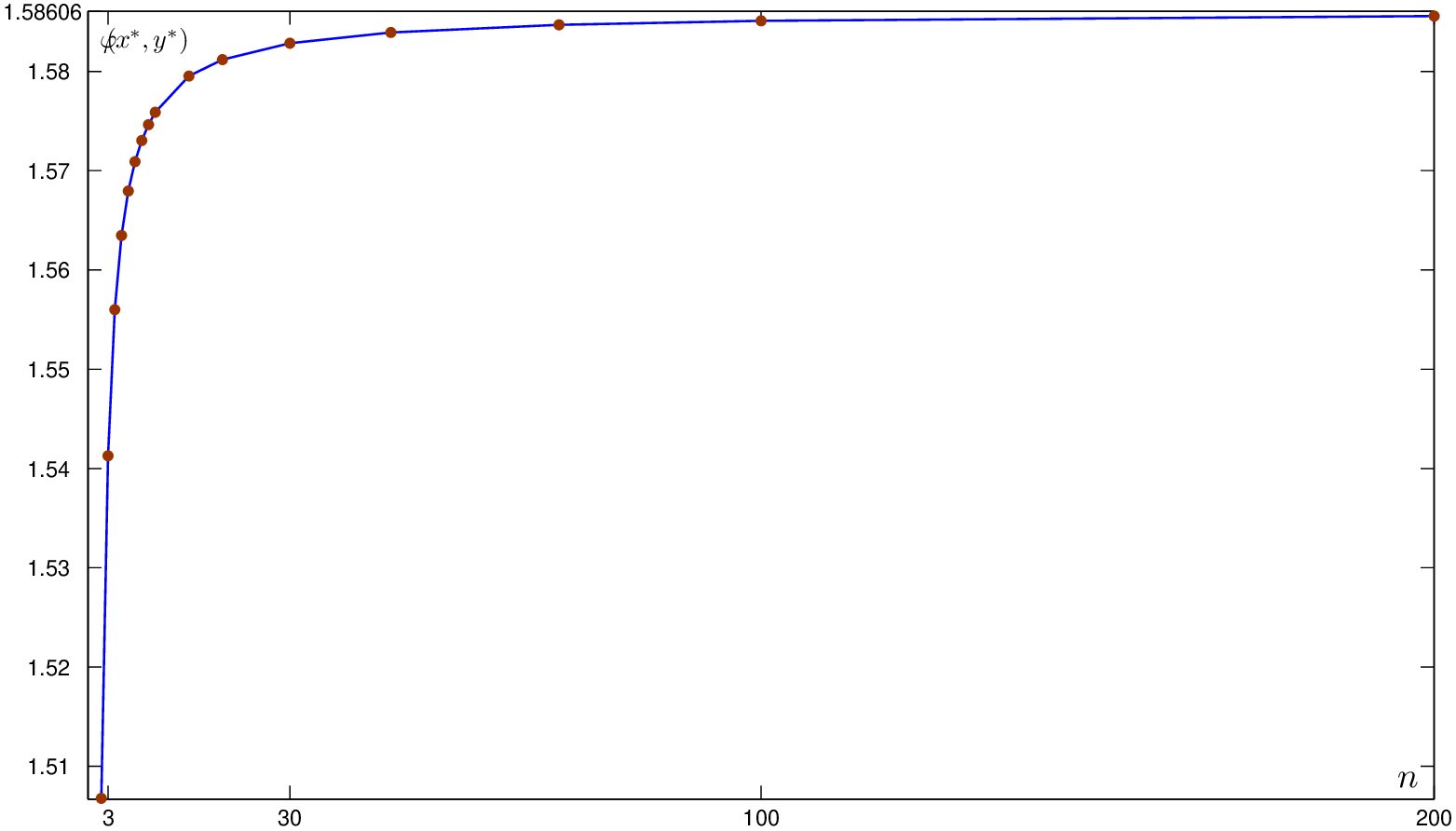}
\vspace{-7mm}\caption{ \label{fg:increase} The approximation ratio $\varphi(x^*,y^*)$ is increasing in the number $n$ of tasks. }
\includegraphics[scale=0.45]{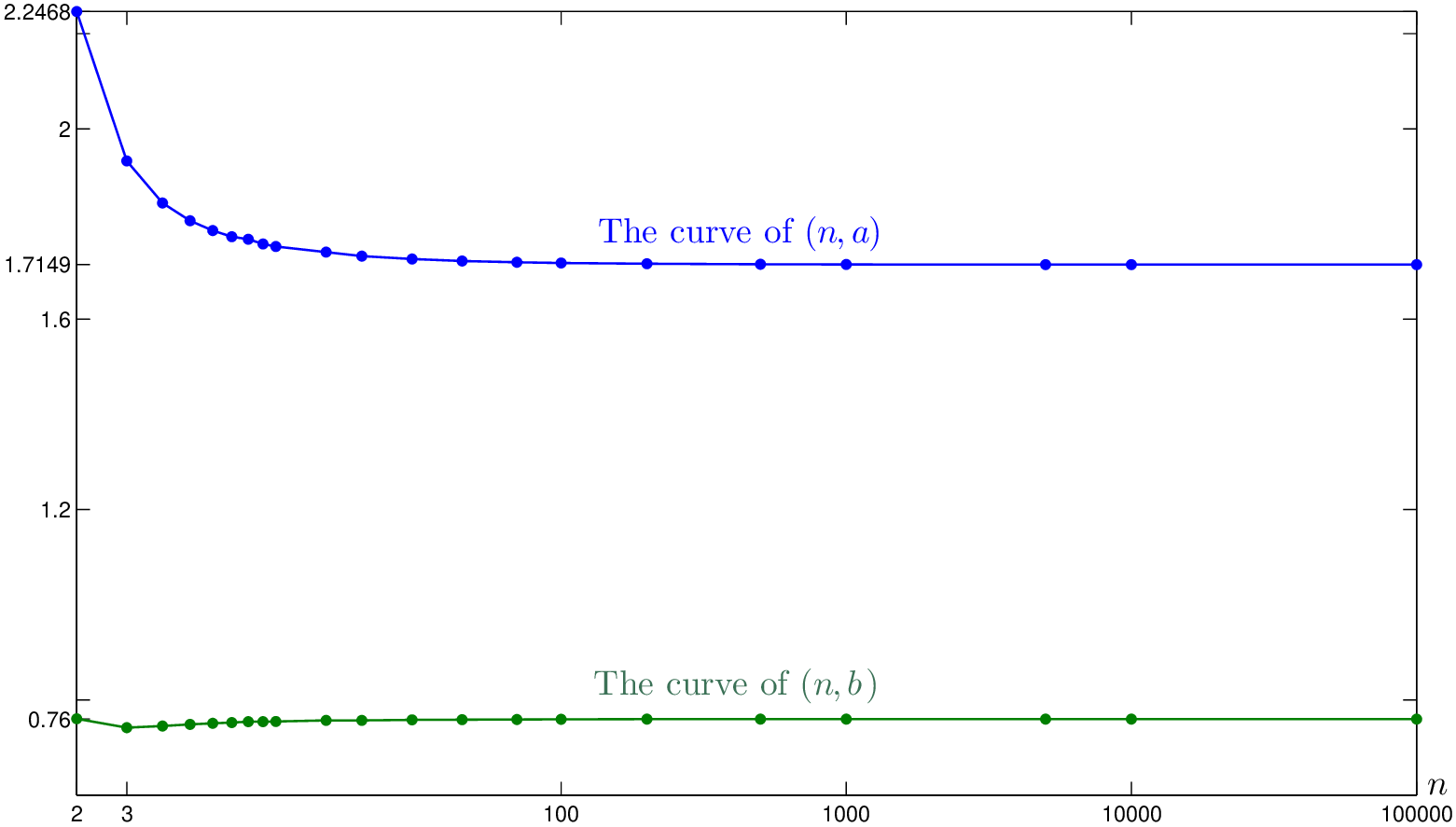}
\vspace{-7mm}\caption{ \label{fg:ab} The optimal value of $a$ (resp. $b$) is decreasing (resp. increasing) in the number $n$ of tasks (when $n\ge3$), and approaches  $1.715$ (resp. 0.76) as $n\rightarrow\infty$. }
\end{center}
\end{figure}

More interesting phenomena  are observed from the first three columns of Table \ref{tb2}: the optimal value of $a$ decreases with $n$, and approaches a limit approximately equal 1.7149; while starting from $n=3$ the optimal value of $b$ increases with $n$, and approaches a limit approximately equal 0.7599. See Figure \ref{fg:ab} for an illustration. Note that the limits ``coincide'' with the setting $a=1.715$ and $b=0.76$ in Theorem \ref{th:independent} for strongly independent tasks. The reason is that for any distribution function $F(\cdot)$, function $\varphi(x,y)$ in (\ref{defphi}) is always upper bounded by function $\varphi(x,y)$ in (\ref{independent}), and the former approaches the latter as $n$ tends to infinity. This fact is implied by Lemma \ref{approach} below.

\begin{lemma}\label{approach}
Given any distribution function $F(\cdot)$, it holds that
\begin{align*}
\left[\left(\sqrt[n-1]{F(x)} + \sqrt[n-1]{F(y)} - 1\right)^+\right]^{n-1}&\le F(x)F(y),\\
\lim_{n\rightarrow\infty}\left [ \left ( \sqrt[n-1]{F(x)} + \sqrt[n-1]{F(y)} - 1
\right )^+ \right ]^{n-1} &= F(x)F(y).
\end{align*}
\end{lemma}
\begin{proof}For the first statement, by writing $u=\sqrt[n-1]{F(x)} $ and $v=\sqrt[n-1]{F(y)} $, it suffices to show that $u+v-1\le uv$. Recall that $u,v\in[0,1]$. By $v\le1$ we have $u\ge u+v-1$, in turn by $u\le1$ we obtain $u(1-u)\ge(u+v-1)(1-u)$, which is equivalent to $u+v-1\le uv$.
In particular, this means that the   limit in the second statement exists, and its value follows from
$\lim_{n\to0}(a^n+b^n-1)^{\frac1n} = e^{\lim_{n\to0} \frac{1}{n} \log(a^n+b^n-1)} = e^{\log a + \log b} = ab$.
\end{proof}

\section{Concluding remark}\label{sec:conclude}
We note that the choice of Clayton Copula in (\ref{defG}) is not accidental. We wish to choose the Copula {which} leads to the best approximation ratio for our mechanism. However, Clayton Copula is the best lower bound among all Archimedean Copulas \cite{mcneil2009multivariate}. Therefore, any hope to improve the bounds presented in this work will have to resort to non-Archimedean Copulas, which usually lack the nice closed-form property of Archimedean Copulas.

\medskip\noindent{\bf Acknowledgements.} \quad  The authors are grateful Professor Pinyan Lu for providing the full versions of the conference papers \cite{l09,ly08}.   This work was done while Xujin Chen was visiting Faculty of Business Administration, University of New Brunswick. Xujin Chen was supported in part by NNSF of China (11222109), NSERC grants  (283106, 290377) and  CAS Program for Cross \& Cooperative Team of Science  \& Technology Innovation. Donglei Du was supported in part by NSERC grant 283106.
\bibliography{mechanism}

\appendix

\section*{Appendix}

\section{Details omitted in the proof of Theorem \ref{th:independent}}
\label{apx:independent}
In this section, we provide some routine computations  and elementary observations omitted in the proof of Theorem \ref{th:independent}.

 \paragraph{\sc Case 1.} It follows from (\ref{ourformula}) that $\varphi(x,y)=y+\left(1+\frac1x-y\right)\left(1-\frac{2(1-b)(a-y)}{a-1}\right)$  when $y\in [\frac{a+1}2,a)$, and  $\varphi(x,y)=y+\left(1+\frac1x-y\right)\left(\frac12+\frac{(2b-1)(y-1)}{a-1}\right)$ when $y\in (1,\frac{a+1}2)$.

\paragraph{\sc Case 2.}  Function $1+\frac{F(x)}x$ has positive derivative $\frac{541}{3575x^2}$ in $(\frac{a+1}2,a)$,   $ \frac{5}{22x^2}$ in $(1,\frac{a+1}2)$, $\frac{32-27x}{22x^3}$ in  $[\frac2{a+1},1)$, and $\frac{12( 400-343x)}{3575x^3}
$ in   $x\in(\frac1a,\frac2{a+1})$.

\paragraph{\sc Case 3.1.}
It follows from (\ref{ourformula}) and (\ref{xy>=1}) that $\varphi(x,y)\!=\!1\!+\!y\!-\! \left(1\!\!-\!\!\frac{2(1\!-\!b)(a\!-\!x)}{a-1}\right) -y\left (  1-\frac{2(1-b)(a-y)}{a-1}  \right) +\left(1+\frac1x\right)
 \left(1-\frac{2(1-b)(a-x)}{a-1} \right)\left (  1-\frac{2(1-b)(a-y)}{a-1}  \right)  $ and
\begin{eqnarray*}
2500(a-1)^2x^2\cdot\frac{\partial\varphi}{\partial x}(x,y)&=&(576x^2-129.84)y-987.84x^2-29.2681\,,\\
 \frac{2500(a-1)^2x}{ b-1 }\cdot\frac{\partial\varphi}{\partial y}(x,y)&=&7150xy-2400x^2-7990.125x+541\,.
\end{eqnarray*}
Among the four roots of the biquadratic equation {$ 276480000x^4-492148800x^3+165634183x-14048688=0$,} only one  $x_0=1.5419412254952502713930612434406$ belongs to $[\frac{a+1}2,a)=[1.3572,1.715)$, and the other three 0.6970...,
 0.0866...,
  $-0.5455...$ are much less than $\frac{a+1}2$.

\paragraph{\sc Case 3.2.} It follows from (\ref{ourformula}) and (\ref{xy>=1}) that $\varphi(x,y)\!=\!1\!+\!y\!-\! \left(  \!1\!\!-\!\!\frac{2(1\!-\!b)(a\!-\!x)}{a-1}  \right)  \!-\!y\left(  \frac12+\frac{(2b-1)(y-1)}{a-1}  \right)    +\left(1+\frac1x\right)
 \left(  1-\frac{2(1-b)(a-x)}{a-1}  \right)$, and
\begin{eqnarray*}
2500(a-1)^2x^2\cdot\frac{\partial\varphi}{\partial x}(x,y)&=&(624x^2+140.66)y-1053x^2-43.95625\,,\vspace{1mm}\\
 2\times 10^4(a-1)^2x\cdot\frac{\partial\varphi}{\partial y}(x,y)&=&-14872xy+4992x^2+16414.97x-1125.28\,.
\end{eqnarray*}
If $\frac{\partial\varphi}{\partial x}(x,y)=0=\frac{\partial\varphi}{\partial y}(x,y)$, then
$ \frac{1053x^2+43.95625}{624x^2+140.66}=y= \frac{4992x^2+16414.97x-1125.28}{14872x}$, giving
$92160000x^4-160274400x^3+48970779x-4682896=0$.
Among the  four roots of the biquadratic equation, only  $x_0\!=\!\!1.5249070327751520068531319284494$ belongs to $(\frac{a+1}2,a)$, the other three $0.642...$, $0.098...$,
  $-0.526...$ are less than 0.7. Thus $\varphi(x,y)$ has a unique critical point $(x_0, \frac{1053x_0^2+43.95625}{624x_0^2+140.66} )
$ when  $x\in(\frac{a+1}2,a)$ and $y\in(1,\frac{a+1}2)$. If $y=1$ and $x\in(\frac{a+1}2,a)$, then the derivative of  $\varphi(x,1)$ is $\frac{541}{7150x^2}-\frac{48}{143}< \frac{541 }{7150 }-\frac{48}{143}<0$, implying that $(x^*,y^*)$ does not belong to this case.

\paragraph{\sc Case 3.3.} It follows from (\ref{ourformula}) and (\ref{xy>=1}) that $\varphi(x,y)=1+y- \left(  1-\frac{2(1-b)(a-x)}{a-1}  \right)  -y\left(  \frac12-\frac{(2b-1)(1/y-1)}{a-1}  \right)   +\left(1+\frac1x\right)
 \left(  1-\frac{2(1-b)(a-x)}{a-1}  \right)\left(  \frac12-\frac{(2b-1)(1/y-1)}{a-1}  \right) $, and
\begin{eqnarray*}
 2\times10^4(a-1)^2x^2y\cdot\frac{\partial\varphi}{\partial x}(x,y)&=&(1560y-4992)x^2+1898.91y-1125.28 \,.
\end{eqnarray*}
If $\frac{\partial\varphi}{\partial x}(x,y)=0$, then $(1560y-4992)x^2+1898.91y-1125.28 =0$ implies $x^2
=\frac{5.41}{24}\left(\frac{70.4}{16-5y}-5.4\right) $.

\paragraph{\sc Case 3.4.} It follows from (\ref{ourformula}) and (\ref{xy>=1}) that $\varphi(x,y)=1+y- \left(  1-\frac{2(1-b)(a-x)}{a-1}  \right)  -y\left(   \frac{2(1-b)(a-1/y)}{a-1}   \right)    +\left(1+\frac1x\right)
 \left(  1-\frac{2(1-b)(a-x)}{a-1}  \right) \left(   \frac{2(1-b)(a-1/y)}{a-1}   \right)  $, and
\begin{eqnarray*}
 \frac{2500(a-1)^2x^2y}{b-1}\cdot\frac{\partial\varphi}{\partial x}(x,y)=(2400-541y)x^2-927.815y+541\,.
\end{eqnarray*}
If $\frac{\partial\varphi}{\partial x}(x,y)=0$, then $(2400-541y)x^2-927.815y+541=0$ implies $x^2
=\frac{3575}{2400-541y}-1.715
<\frac{3575}{2400-541}-1.715<0.21$, contradicting the hypothesis $x\in[\frac{a+1}2,a)$ of Case 3. Thus $\frac{\partial\varphi}{\partial x}(x,y)\ne 0$, and it suffices to consider the case where $x=\frac{a+1}2$.
Note that the derivative of $\varphi(\frac{a+1}2,y)$ is $\frac{573344}{647075y^2}-\frac{541}{3575}>\frac{573344}{647075}-\frac{541}{3575}>0$. We deduce that  $(x^*,y^*)$ does not belong to Case 3.4.

\paragraph{\sc Case 4.} In case of $y\in[\frac{a+1}2,a)$, we have 
$\varphi(x,y)=1+y- \left(  \frac12+\frac{(2b-1)(x-1)}{a-1} \right) -y\left (  1-\frac{2(1-b)(a-y)}{a-1}  \right)   +\left(1+\frac1x\right)
 \left(  \frac12+\frac{(2b-1)(x-1)}{a-1} \right) \left (  1-\frac{2(1-b)(a-y)}{a-1}  \right) $, and
\begin{eqnarray*}
62500(a-1)^2x^2\cdot\frac{\partial\varphi}{\partial x}(x,y)&=&(15600x^2+4875)y-26754x^2- \frac{35165}{32},
\vspace{0.5mm} \\
 \frac{200(a-1)^2x}{b-1}\cdot\frac{\partial\varphi}{\partial y}(x,y)&=& 572xy-208x^2-633.49x+65.
\end{eqnarray*}
If $\frac{\partial\varphi}{\partial x}(x,y)=0=\frac{\partial\varphi}{\partial y}(x,y)$, then
$\frac{26754x^2 +{35165}/{32}}{15600x^2+4875}=y=\frac{208x^2+633.49x-65}{572x}$, which implies $153600x^4-256608x^3+116435x-15000=0$. Among the four real roots of this biquadratic equation, only one root  $x_0=1.20271213592780899067095170...
$ belongs to  $I_3=[1,\frac{a+1}2)$, the other three  0.964..., 0.133..., $-0.629...$ are less than 1. Thus when $x\in(1,\frac{a+1}2)$ and $y\in(\frac{a+1}2,a)$, function $\varphi(x,y)$ has a unique  critical point of  $(x_0, \frac{26754x_0^2 +{35165}/{32}}{15600x_0^2+4875})%
$. {If $x=1$ and $y\in(\frac{a+1}2,a)$, then $\frac{\partial\varphi}{\partial y}(1,y)<\frac{b-1}{200(a-1)^2}(572\frac{a+1}2-208-633.49+65)=0$.} {If  $y=\frac{a+1}2$ and $x\in(1,\frac{a+1}2)$,  $\frac{\partial\varphi}{\partial x}(x,\frac{a+1}2)= \frac{19}{110x^2}-\frac{48}{275}<\frac{19}{110 }-\frac{48}{275}<0$.}

In case of $y\in[1,\frac{a+1}2)$, it follows from (\ref{ourformula}) and (\ref{xy>=1}) that
$\varphi(x,y)=1+y- \left(  \frac12+\frac{(2b-1)(x-1)}{a-1} \right) -y\left(  \frac12+\frac{(2b-1)(y-1)}{a-1}  \right)   +\left(1+\frac1x\right)
 \left(  \frac12+\frac{(2b-1)(x-1)}{a-1} \right) \left(  \frac12+\frac{(2b-1)(y-1)}{a-1}  \right) $, and
\[
10^4(a-1)^2x^2\cdot\frac{\partial\varphi}{\partial x}(x,y)=(2704y-4563)x^2+845y-\frac{4225}{16}\,.
\]
If $\frac{\partial\varphi}{\partial x}(x,y)=0$, then the above equation implies $x^2=\frac{6.875}{27-16y}-\frac5{16}
$.

In case of $y\in\frac2{a+1},1)$,  we have 
$\varphi(x,y)=1+y- \left(  \frac12+\frac{(2b-1)(x-1)}{a-1} \right) -y\left(  \frac12-\frac{(2b-1)(1/y-1)}{a-1}  \right)   +\left(1+\frac1x\right)
 \left(  \frac12+\frac{(2b-1)(x-1)}{a-1} \right) \left(  \frac12-\frac{(2b-1)(1/y-1)}{a-1}  \right) $, and
 \begin{eqnarray*}
32000(a-1)^2x^2y\cdot\frac{\partial\varphi}{\partial x}(x,y)=(2704y-8652.8)x^2+4563y-2704\,.
\end{eqnarray*}
If $\frac{\partial\varphi}{\partial x}(x,y)=0$, then the above equation implies $x^2=\frac{80-135y}{80y- 256}$.

\paragraph{\sc Case 5.} In case of $y\in[\frac{a+1}2,a)$, we have 
$\varphi(x,y)=1+y- \left(  \frac12-\frac{(2b-1)(1/x-1)}{a-1}  \right)
 -y\left (  1-\frac{2(1-b)(a-y)}{a-1}  \right)
 +\left(1+\frac1x\right)
\left(  \frac12-\frac{(2b-1)(1/x-1)}{a-1}  \right)
\left (  1-\frac{2(1-b)(a-y)}{a-1}  \right)$, and
\begin{eqnarray*}
2500(a-1)^2x^2\cdot\frac{\partial\varphi}{\partial y}(x,y)=-1716x^2y+2524.47 x^2+429x-624\,.
\end{eqnarray*}
If $\frac{\partial\varphi}{\partial y}(x,y)=0$, then the above equation implies $y=\frac{2524.47 x^2+429x-624}{1716x^2}
$.

In case of $y\in(1,\frac{a+1}2)$, we have 
$\varphi(x,y)=1+y- \left(  \frac12-\frac{(2b-1)(1/x-1)}{a-1}  \right)
 -y\left(  \frac12+\frac{(2b-1)(y-1)}{a-1}  \right)    +\left(1+\frac1x\right)
\left(  \frac12-\frac{(2b-1)(1/x-1)}{a-1}  \right)
\left(  \frac12+\frac{(2b-1)(y-1)}{a-1}  \right)    $, and
\begin{eqnarray*}
10^4  (a-1)^2x^3\cdot\frac{\partial\varphi}{\partial x}(x,y)&=&(5408-1859x)y-\frac{50193}{16}x -1690\,,\vspace{1mm}\\
10^4  (a-1)^2x^2\cdot\frac{\partial\varphi}{\partial y}(x,y)&=&-7436x^2y+ \frac{86697}8x^2 +1859x-2704\,.
\end{eqnarray*}
If  $\frac{\partial\varphi}{\partial x}(x,y)=0=\frac{\partial\varphi}{\partial y}(x,y)$, we have $ \frac{50193x/16+1690}{5408-1859x}=y=\frac{86697x^2/8+1859x-2704}{7436x^2}$, which implies $12177 x^3-11928 x^2-4224 x+4096=0$.  Among the three real roots of the cubic equation, only one root $x_0=0.98503501986004557612380063958196
$ belongs to $[\frac2{a+1},1)$, the other two $-0.587...$ and $ 0.581...$ are less than $0.59<0.73<\frac2{a+1}$. Hence, when $x\in(\frac2{a+1},1)$ and $y\in(1,\frac{a+1}2)$, function $\varphi(x,y)$ has a  unique critical point $(x_0,\frac{50193x_0/16+1690}{5408-1859x_0})
$.

\paragraph{\sc Case 6.} It follows from (\ref{ourformula}) and (\ref{xy>=1}) that $\varphi(x,y)=1+y- \left(  \frac{2(1-b)(a-1/x)}{a-1}  \right)
 -y\left (  1-\frac{2(1-b)(a-y)}{a-1}  \right)
+\left(1+\frac1x\right) \left(  \frac{2(1-b)(a-1/x)}{a-1}  \right)
\left (  1-\frac{2(1-b)(a-y)}{a-1}  \right)    $. Note that
\begin{eqnarray*}
 \frac{1250 (a-1)^2x^3}{b-1}\cdot\frac{\partial\varphi}{\partial x}(x,y)=\left(\frac{637637}{400}+858y\right)x-2400y+541\,,
\end{eqnarray*}
saying that $\frac{\partial\varphi}{\partial x}(x,y)$ is  continuous. If $\frac{\partial\varphi}{\partial x}(x,y)=0$, then   the above equation implies  $x=\frac{2400y-541}{637637/400+858y}
$.

\section{Details omitted in the proof of Theorem \ref{th:limit}}
$\lambda_2=\frac{169(3830\lambda_1-2367)}{2500(294\lambda_1-169)} =\frac{18069396176}{32281745375}+\frac{2054533\sqrt{154595269}}{129126981500}$,
$\lambda_3=\frac{169}{250}-\frac{227/2500}{2\lambda_2-1}=\frac{17274798609857}{22964052192750}-\frac{466378991\sqrt{154595269}}{22964052192750}$,
$\lambda_4=\frac12-\frac{0.0908}{2\lambda_3-0.648}=
=\frac{7635461853+9926731\sqrt{154595269}}{137555892260+32555020\sqrt{154595269}}$.

\section{More accurate data for Table \ref{tb2}}\label{apx:data}

 \begin{sidewaystable}[h!]
\begin{center}
\renewcommand{\arraystretch}{1.45}
{
\begin{tabular}{c|   c| c|c}
\hline
$n$ &   $x^*$  &$y^*$   &$ \varphi(x^*,y^*)$  \\ [0ex] %
 \hline

2&    $  1.6234 $  &  1.9313955485585601046238934941357
& 1.5067710963980944782747428689618
\\[0ex]
\hline

3&     1.9105670668253638133649019437144
    &  1.7231009559709047351816479931585
    &1.5412707360547943657991254440276
    \\[0ex]
\hline

 4& 1.8442 &  1.6932202822890392024390848746407
    &1.5559952304614046436626040303963
    \\ [0ex]
\hline

5   &  1.1418758035530052197259465174284
&  1.5193285943718712882599675140227
 &1.5634859374811611587574589066207
 \\[0ex]
\hline

6&1.1468400067157940025452944610151
    & 1.4989121029040246568797556392383
      &1.5679473463485327222599607921438
       \\[0ex]
\hline

7&  1.1447309125170275212468595782411
  &   1.484571582878536188943030538212 
 &1.5709131850723250245494000409963
 \\ [0ex]
\hline

8& 1.1192295299099999095204793775338
& 1.4661575387460290542662733059842
&1.5730320736692182670424244861351
\\ [0ex]
\hline

9&  1.37905& 1.5499380480779130220270189965959
&1.5746303803351011652011948172003
\\ [0ex]
\hline

10&  1.0673757071298466403419524795027
& 1.4334673997356224273147518033511
&1.5758769994650307921801868360490
\\ [0ex]
\hline

15&  1.0190835924366512532657225165167
& 1.3975512391926399047292761679273
&1.5795353026978935506718926262693
\\ [0ex]
\hline

20& 0.99970778780101732241547551893746
& 1.3798783532170473264955035119783
&1.5811826689588861505342265445506
\\ [0ex]
\hline

30& 1.36335
& 1.5259350403311591204413844025112
&1.5828322597883834887966258975212
\\ [0ex]
\hline

45& 0.98794524618663881465607801146689
& 1.3491108560548288330949162627803
&1.5839252560845547002088551380439
\\ [0ex]
\hline

70& 0.98688203426265674877981837198604
&1.3445069231326205461130030016648
&1.5846893836898565677273609253461
\\ [0ex]
\hline

100& 1.35915
& 1.5197905945463969779041235597106
&1.5850948284784656117096801608568
\\ [0ex]
\hline

200& 1.35835
& 1.5186228330081581461286077683326
&1.5855735652961084891643395167193
\\ [0ex]
\hline

500&0.9851752572294799614738280979509
& 1.3375636313202476923578387868474
&1.5858603199943162032070631539682
\\ [0ex]
\hline

1000& 1.35765
& 1.517614059591700259588264998456
&1.5859488979551645826404637773521
\\ [0ex]
\hline

5000& 0.9849521897949018445217461703578
& 1.3365770912703036632507291869842
&1.5860275919063095972916244136286
\\ [0ex]
\hline

$10^4$& 1.7149
& 1.6490128247935071925667216419242
&1.5860403769478577107321370931459
\\ [0ex]
\hline

$10^5$& 0.98496211975134262406328389261034
&  1.3364590898298425170054315458401
&1.5860442150763098823063046438619
\\ [0ex]
\hline

$10^6$&0.98495134013425345020920076422044
&1.3364514129617508508829359925585
&1.5860456086356999882980289839907
\\ [0ex]
\hline

$\vdots$ &$\vdots$
&$\vdots$
&$\vdots$\\ [0ex]
\hline

$\infty$& 1.3575     &  1.5174263351749539594843462246777
   &1.5860582220359942251519669298432\\ [0ex]
\hline
\end{tabular}
}\end{center}
\vspace{-2mm}\caption{Long digital expressions of data from Table \ref{tb2}, where the values of $a$, $b$ and $\delta$ are omitted.}
\label{tb3}
\end{sidewaystable}

\end{document}